\documentclass[jmp,aps, amsmath,amsthm, amssymb,reprint, showkeys, showpacs]{revtex4-1}

\usepackage{graphicx}
\usepackage{dcolumn}
\usepackage{bm}

\newtheorem{Thm}{Theorem}[section]
\newtheorem{theorem}[Thm]{Theorem}

\newtheorem{corollary}[Thm]{Corollary}
\newtheorem{lemma}[Thm]{Lemma}
\newtheorem{conjecture}[Thm]{Conjecture}
\newtheorem{remark}{Remark}[section]
\newtheorem{definition}[Thm]{Definition}
\newcommand{\beq}{\begin{equation}}
\newcommand{\eeq}{\end{equation}}

\def\supp{{\rm supp\,}}

\newenvironment{proof}[1][Proof]{\begin{trivlist}
\item[\hskip \labelsep {\bfseries #1}]}{\end{trivlist}}

\newcommand{\qed}{\nobreak \ifvmode \relax \else
      \ifdim\lastskip<1.5em \hskip-\lastskip
      \hskip1.5em plus0em minus0.5em \fi \nobreak
      \vrule height0.75em width0.5em depth0.25em\fi}

\def\cA{{\cal A}}

\def\cD{{\cal D}}
\def\cH{{\cal H}}
\def\cO{{\cal O}}

\def\cS{{\cal S}}
\def\cR{{\cal R}}

\newcommand{\RR}{\mathbb{R}}   

\begin{document}

\title{A Goldstone Theorem in Thermal Relativistic Quantum Field Theory }


\author{Christian D.\ J\"akel}
\email{christian.jaekel@mac.com}
\affiliation{School of  Mathematics, Cardiff University, Wales, \\
CF24 4AG, United Kingdom.}

\author{Walter F. Wreszinski}
 \email{wreszins@gmail.com, supported in part by CNPq}
\affiliation{%
Departamento de F\'isica Matem\'atica,  
Instituto de F\'isica, \\
USP, Caixa Postal 66318  
05314-970, S\~ao Paulo, Brazil.}

\begin{abstract}
We prove a Goldstone Theorem in thermal relativistic quantum field theory, which relates 
spontaneous symmetry breaking to the rate of space-like decay of the two-point function. 
The critical rate of fall-off coincides with that of the massless free scalar field theory. Related results and open problems are briefly discussed.
\end{abstract}

\maketitle

\section{Introduction and Summary}

Thermal quantum field theory ({\em tqft}) has received considerable attention recently, both from the conceptual 
and the constructive point of view (see \cite{1} for a review and references). Its range of applications extends from heavy ion collisions and cosmology at early stages (see~\cite{2} for a review) to the present day hot big-bang model in cosmology \cite{3}, with obvious potential relevance to the dark energy problem \cite{4}, which, however, remains to 
be explored. 

In the present paper, we study the spontaneous symmetry breaking ({\em ssb}) of continuous (internal) symmetries in relativistic thermal quantum field theory, and prove a version of Goldstone's theorem (see, e.g., \cite{5}
for a review and references) --- Theorem \ref{Th1} of Section~\ref{III} --- which relates {\em ssb} to the asymptotic decay of 
(truncated) correlation functions for large space-like distances. In this respect the theorem follows the lines of 
\cite{6} and \cite{7}, the latter having been proved to be an optimal version, generalizing the well-known 
Mermin-Wagner theorem of quantum statistical mechanics \cite{8}. If, however, one endeavors to understand
the concept and structure of particles in {\em tqft}, large time-like distances necessarily come into play, and in this 
connection the Goldstone-type theorem of Bros and Buchholz \cite{9} is more natural (see Remark \ref{Rm1} of 
Section~\ref{III} and Section~\ref{IV} --- discussion and outlook). 

The main advantage of our approach lies in the possibility of a sharp distinction between massive and 
zero-mass theories in terms of their correlation functions' rate of space-like decay (Conjecture \ref{Con1} 
of Section~\ref{III}):
if such is true, a theorem of the same form as the vacuum (T=0, zero density) version \cite{10a, 10b} follows 
(Corollary \ref{Cor1}) and Theorem \ref{Th1} turns out to be optimal as in the quantum statistical mechanical case. 

Our proof of Theorem \ref{Th1} generalises the method used in \cite{6, 7}, which was based on the Bogoliubov
inequality (see \cite{11} and the references given there), in an essential way: firstly, the treatment of the 
middle term in that inequality (see (82) of Appendix A) relies on local current conservation, Einstein causality and the definition (\ref{8}) of the global charge, in a manner reminiscent of \cite{12}; secondly, unlike in \cite{6, 7}, 
we employ a form of the Bogoliubov inequality which was proved to follow from the KMS condition for infinite systems
by Garrison and Wong \cite{13} in a $C^*$-algebraic framework. This  naturally takes into account the singular nature 
of the quantum fields, which reflects itself in the necessity of choosing adequate test functions. For other related derivations, 
see \cite{14} and \cite{15}, and \cite{16}, Vol.~II, pg. 334.

We now briefly describe the organisation of the paper. In Section~\ref{II} we introduce the framework, which is that of 
$C^*$-dynamical systems (see e.g.~\cite{24}) and formulate our assumptions, together with some auxiliary lemmas.
In Section~\ref{III} we prove our main result (Theorem~\ref{Th1}), followed by Conjecture~\ref{Con1} and 
Corollary~\ref{Cor1} referred to before. The connection to $W^*$-dynamical systems and the spectral properties of the Liouvillian is also discussed there. 
Section~\ref{IV} is reserved for a discussion and outlook, in particular the relation to other approaches and open problems.
In Appendix~A we state Bogoliubov's inequality, with some additions needed in the main text. In Appendix~B we state, for the reader's convenience, the theorem on the partition of unity used in the main text. 

\section{Framework and Assumption}
\label{II}

We work in the framework of $C^*$-dynamical systems, 
consisting of a pair $(\cA, \tau)$, where $\cA$ is a $C^*$-algebra with unit and $\{ \tau_t \}_{t \in \RR}$ is a
one-parameter group of norm continuous (time translation) automorphisms of $\cA$ (see, e.g., \cite{24}). 
Since the time-translation automorphisms  are not norm continuous
on the Weyl algebra (see \cite{16}, Vol.~2, Theorem 5.2.8), we adopt Haag's construction (see \cite{17}, pp.~129
et seq.), which leads to the following structure (see \cite{17}, III.3.3., pg.~141) - $\cO$ denotes a finite, contractible, 
open region in Minkowski space: 

\begin{itemize}
\item [(i)] a net of $C^*$-algebras with common unit 
	\beq
		\label{1}
		\cO \to \cA_S (\cO)
	\eeq
with the total $C^*$-algebra (the $C^*$-inductive limit \cite{18}) $\cA_S$:
	\beq
		\label{2}
			\cA_S = \overline {\cA_L} , 
				\qquad 
		\hbox{with} \quad  \cA_L = \bigcup_{\cO} \cA_S ( \cO) ,
	\eeq
where the bar denotes the completion in the norm topology. We call $\cA_L$ the (strictly) local algebra. The 
action of the time automorphism $\{ \tau_t \}_{t \in \RR}$ on $\cA_S$ is $t$-continuous in the norm topology (in \cite{17}, this 
is required for the space-time automorphisms, but we do not need the smoothness with respect to spatial translations); 

\item [(ii)] a set $\cS$ of physical (i.e., locally normal \cite{17, 18}) states over $\cA_S$ and the complex linear span of $\cS$, denoted by $\Sigma$; 

\item [(iii)] the dual of $\Sigma$ is a net of $W^*$-algebras with common unit 
\beq
\label{3}
\cO \mapsto \cR(\cO) = \Sigma(\cO)^* .
\eeq
$\cR(\cO)$ is closed in the weak topology induced by $\Sigma$ and $\cA_S (\cO)$ is weakly dense in $\cR (\cO)$. 

\item [(iv)] Local commutativity: if the regions $\cO_1$ and $\cO_2$ are totally space-like to one another, then 
\beq
\label{4}
[ A , B] = 0  \qquad \forall A \in \cA(\cO_1), \quad \forall B \in \cA(\cO_2).
\eeq
\end{itemize}
We must also assume certain global properties on the particular state $\omega \in \cS$ we shall work with. 
The basic property of thermal states is the KMS condition (see \cite{16}, Vol.~2):

\begin{definition} A state $\omega $ $(= \omega_\beta)$ over $\cA $ $(=\cA_S)$ is called a KMS state for some $\beta >0$, 
if for all $A, B \in \cA$ there exists a function $F_{A, B}$,  which is continuous in the strip $0 \le \Im z \le \beta$
and analytic and bounded in the open strip $0 < \Im z < \beta$, with boundary values given by 
\beq
\label{5}
F_{A, B} (t) = \omega(A \tau_t (B)) 
\eeq
and
\beq
F_{A, B} (t+i \beta) = \omega( \tau_t (B) A) 
\eeq
for all $t \in \RR$. 
\end{definition}

We further assume that

\begin{itemize}
\item [{\bf A1}] $\omega$ is a factor (primary) state over $\cA_S$; 
\item [{\bf A2}] $\omega$ satisfies the KMS condition.
\end{itemize}

From {\bf A2} it follows that $\omega$ is invariant under time translations, but we also need that

\begin{itemize}
\item [{\bf A3}] $\omega$ is invariant under space translations. 
\end{itemize}

By {\bf A2}, {\bf A3} and the GNS construction there exists a representation $\pi_\omega$ of $\cA_S$ on a Hilbert space $\cH_\omega$, with self-adjoint space-time translation generators $(L_\omega, \vec P_\omega)$ and cyclic vector 
$\Omega_\omega$ such that
\beq
\label{6ab}
L_\omega \Omega_\omega = 0 ,
\eeq
and
\beq
\label{6c}
\vec P_\omega \Omega_\omega = \vec 0 . 
\eeq

As occurs with $W^*$-dynamical systems, $L_\omega$ is not bounded below, see (74) et seq..
Of primary concern to us will be continuous internal symmetries generated by a local current~$J_\mu (x)$ on which we make
the same assumptions as in \cite{10a} (see p.~110), headed there under local current conservation. Before stating them we remark that in the following, when $A \in \cA_L$ occurs in connection with the 
representation $\pi_\omega$, it is understood as $\pi_\omega(A)$. The assumptions are:  there exist for 
every test-function $f \in \cD = C_0^\infty (\RR^{s+1})$ a set of  $(s+1)$~unbounded self-adjoint operators $J_\mu (f)$, with the properties 

\begin{itemize}
\item [{\bf A4}] $\Omega_\omega$ is in the domain of $J_\mu (f)$  for all $f \in \cD$;
\item [{\bf A5}] $T(a) J_\mu (f) T(a)^{-1} = J_\mu (f_a)$ where $f_a (x) = f (x-a)$; 
\item [{\bf A6}] $\sum_{\mu = 0}^s J_\mu \bigl( \frac{\partial f}{\partial x_\mu} \bigr) = 0$;
\item [{\bf A7}] (a) $(\Omega_\omega , [ J_\mu (f) , A ] \Omega_\omega) = 0$ for $A \in \cA_S (\cO)$, if
the support of $f$ is totally space-like to~$\cO$;

\item [] (b)
$(\Omega_\omega , [ J_0 (f), \vec J (g) ] \Omega_\omega) =\vec 0 $, 
if the supports of $f$ and $g$ are space-like to one another; 
\item [{\bf A8}]
for all $f \in \cD$, the charge operator $J_0 (f)$ is affiliated to $\cR (\cal{O})$.
\end{itemize}

\noindent
In {\bf A7} (a) the natural definition 
\begin{eqnarray}
\label{7}
(\Omega_\omega , [ J_\mu (f) , A ] \Omega_\omega) &= &(J_\mu (f) \Omega_\omega , A \Omega_\omega) 
\nonumber \\
&& \qquad - (A^* \Omega_\omega , J_\mu (f) \Omega_\omega)
\end{eqnarray}
is adopted. By {\bf A4} and (\ref{7}), the term
\beq
(\Omega_\omega , [ J_0 (f), \vec J (g) ] \Omega_\omega)
\eeq
is well-defined for all $f, g \in \cD$. {\bf A7}(b) follows from the assumption e.)~of \cite{10a} (see pg.~110) that 
$(\Omega_\omega , [ J_0 (f), \vec J (g) ] \Omega_\omega)$ is a tempered distribution, but we only need {\bf A7}(b). 

Assumption {\bf A8} had to be imposed on $J_\mu (f)$, because the KMS condition (\ref{5})  a priori holds
only for $A, B \in \cA$. Recall that $\cR(\cO)$ is the von Neumann algebra defined in (\ref{3}), and the concept of affiliation is defined in 
\cite{16} (see Vol.~1, Definition 2.5.7, pg.~87). By self-adjointness of $J_0 (f)$ and \cite{16}, Lemma~2.5.8 (see Vol.~1, pg.~87) the spectral projections $E(\lambda)$ of $J_0 (f)$ lie in $\cR (\cO)$. Note that it is too
much to require that they lie in $\cA_S (\cO)$: bounded functions of the fields are expected to belong {\it only} to the 
weak closure of the Weyl algebra, and are thus not smooth elements. 

We have now completed our list of assumptions, and turn to our criterion of existence of~{\em ssb}, which is the same as the one adopted in \cite{10a}. One might expect that the limit $V \to \infty$ of the local integrated current density 
\beq
\label{charge}
\int_V {\rm d}^s \vec x \; J_0 (x_0, \vec x)  
\eeq
defines a global charge operator, which serves as the generator of the internal symmetry considered. 
However, the limit $V \to \infty$ of (\ref{charge}) does not exist due to vacuum fluctuation occurring all over space, by translation
invariance: this is, in fact, as remarked in \cite{12}, the source of {\em ssb}. We therefore
define, as in \cite{10a}, the charge operator corresponding to $J_\mu$ as a suitable limit of the operators
\beq
\label{8}
J_0 (f_d \otimes g_R) := 
\int {\rm d}^{s+1} x \; f_d (x_0) \, g \left( \frac{\vec x }{R} \right) J_0 (x)
\eeq

as $R \to \infty$, where
\beq
\label{9a}
g \in \cD_s := C_0^\infty (\RR^s),
\eeq
\beq
\label{9b}
f_d \in \cD := C_0^\infty (\RR),
\eeq
\beq
\label{9c}
\int {\rm d} x_0 \; f_d (x_0) = 1,
\eeq
\beq
\label{9d}
 f_d (x_0)=0 \quad \hbox{if} \quad |x_0| \ge d ,
 \eeq
\beq
\label{9e}
g (\vec x)=1 \quad \hbox{if} \quad |\vec x| \le 1 ,
\eeq
\beq
\label{9f}
g (\vec x)=0 \quad \hbox{if} \quad |\vec x| > 1 + \delta , \quad 0 < \delta < 1  ,
\eeq
with $\vec J (f \otimes g_R)$ defined similarly. For other choices, see ref. [35]. The symmetry is characterized by the following property (see \cite{10a}, p.~111):
there exists a one-parameter group of automorphisms $A \mapsto A^\lambda$ of $\cA_S$, strongly continuous with respect to $\lambda$ , such that:

\begin{itemize}
\item [(a)] if
$\cO_L = \{ x \in \RR^{s+1}\mid | \vec x| + | x_0| < L \}$, $ L >0$, then 
\beq
A_L \in \cA_S (\cO_L) \quad \hbox{implies} \quad A_L^\lambda \in \cA_S (\cO_L) ;
\eeq
\item [(b)] if $f_d$, $g$ satisfy (\ref{9a})--(\ref{9f}) and $J_0 (f_d \otimes g_R)$ is defined by (\ref{8}), then 
\begin{eqnarray}
\label{10}
&&
\frac{{\rm d} }{ {\rm d} \lambda} (\Omega_\omega , A^\lambda \Omega_\omega) \Bigl|_{\lambda =0} 
\nonumber \\ [2mm]
&& \qquad 
= i \lim_{R \to \infty} (\Omega_\omega , [ J_0 (f_d \otimes g_R), A ] 
\Omega_\omega).
\end{eqnarray}
\end{itemize}

\begin{lemma}
About the limit on the r.h.s.~of (\ref{10}) we have 
\label{Lm1}
\begin{eqnarray}
\label{11}
&& \lim_{R \to \infty} (\Omega_\omega , [ J_0 (f_d \otimes g_R), A_L ] 
\Omega_\omega)
\nonumber \\ [2mm]
&& \qquad = 
(\Omega_\omega , [ J_0 (f_d \otimes g_{R_1}), A_L ] 
\Omega_\omega)
\end{eqnarray}
if 
\beq
\label{12}
R_1 \ge L + d +1.
\eeq
The r.h.s.~in (\ref{11}) is independent of the functions $f_d$ and~$g$, as long as they satisfy (\ref{9a})--(\ref{9f}).
\end{lemma}

\begin{proof}
This follows from local commutativity {\bf A7}(a), see 
Lemma 1, pg.~112, of \cite{10a}. 
\qed
\end{proof}

\noindent
{CRITERION.} Given (\ref{10}), we adopt as in \cite{10a} as our criterion for spontaneous breakdown of the symmetry ({\em ssb}) associated to the one-parameter group of automorphisms of~$\cA_S$: 
\beq
\label{13}
\lim_{R \to \infty} (\Omega_\omega , [ J_0 (f_d \otimes g_R), A_0 ] 
\Omega_\omega)
= c \ne 0 
\eeq
for some $A_0 \in \cA_L$. 

\bigskip

In this paper we shall assume that (\ref{13}) holds, and derive some constraints from it. 

\bigskip

There are two preliminary steps, which we shall prove in this section: firstly, using time translation invariance, 
which follows from {\bf A2}, we show that we may replace $J_0 (f_d \otimes g_R)$ in (\ref{13}) by a smoothened
version, essential to the application of Bogoliubov's inequality in Section~\ref{III}. Secondly, using {\bf A3} we show that 
$A_0$ may be replaced by an average over space-translations, which is essential to relate (\ref{13}) to the rate of 
space-like decay of correlations.

\begin{lemma} 
\label{Lm2}
Assume {\em ssb} takes place, in the sense that (\ref{13}) holds. Then there exists $h \in \cD ({\mathbb R})$ such that (\ref{13}) also holds (with a different $c \ne 0$)
for the observable
\beq
\label{14}
\widetilde A_0 := \int {\rm d} t \; h(t) \tau_t (A_0) .
\eeq
\end{lemma}

\begin{proof}
By Lemma \ref{Lm1}
\begin{eqnarray}
\label{15}
&\lim_{R \to \infty} (\Omega_\omega , [ J_0 (f_d \otimes g_R), A_0 ] 
\Omega_\omega)
\nonumber \\ [2mm]
& \qquad =
 (\Omega_\omega , [ J_0 (f_d \otimes g_{R_1}), A_0 ] 
\Omega_\omega)
\end{eqnarray}
if
\beq
\label{16}
R_1 = L + d +1.
\eeq
Now write (\ref{15}) as
\beq
\label{17}
(J_0 (f_d \otimes g_{R_1}) \Omega_\omega , A_0   \Omega_\omega)
-
(A_0^* \Omega_\omega , J_0 (f_d \otimes g_{R_1})  \Omega_\omega).
\eeq
By (\ref{16}), (\ref{17})  and {\bf A4}, and the norm-continuity of the time evolution 
$\tau_t$ as assumed in (i), given $\epsilon >0$,  we may 
choose $h_\epsilon \in \cD$ such that 
\goodbreak
\begin{eqnarray}
\label{18}
& \kern -.5cm
\left|  (\Omega_\omega , [ J_0 (f_d \otimes g_R), A_0 ] 
\Omega_\omega) \right. \qquad 
\nonumber \\ [2mm]
& \qquad \qquad \left. -
 (\Omega_\omega , [ J_0 (f_d \otimes g_{R_1}), \widetilde A_0 ] 
\Omega_\omega) \right| < \epsilon .
\end{eqnarray}

See \cite {19}, Theorem 4.8, for the proof of (\ref{18}). The identities
(\ref{15}) and (\ref{18}) imply
\begin{eqnarray}
\label{19}
& \lim_{R \to \infty}  (\Omega_\omega , [ J_0 (f_d \otimes g_R), \widetilde A_0 ] 
\Omega_\omega)
\nonumber \\ [2mm]
& \qquad \qquad = (\Omega_\omega , [ J_0 (f_d \otimes g_{R_1}), \widetilde A_0 ] 
\Omega_\omega) = c_\epsilon \ne 0 \; . \quad
\end{eqnarray}
\qed
\end{proof}

In the following we ommit the suffixes $\epsilon $ in $h_\epsilon$ and $c_\epsilon$.

\begin{lemma}
\label{Lm3} 
Assume {\em ssb} takes place, in the sense that~(\ref{13}) holds.  Then
\begin{eqnarray}
\label{20}
&& \kern - 1cm
 \lim_{R \to \infty}  (\Omega_\omega , [ J_0 (f_d \otimes g_R), \widetilde A_0 ] 
\Omega_\omega)
\nonumber \\
&=&
 (\Omega_\omega , [ I_h (f_d \otimes g_{R_1}), A_0 ] 
\Omega_\omega) 
\nonumber \\
&=& c  \ne 0  ,
\end{eqnarray}
where 
\beq
\label{21}
I_h  (f_d \otimes g_R) :=
\int {\rm d} t \; h(t)  \tau_{-t} \bigl( J_0 (f_d \otimes g_{R}) \bigr).
\eeq
$h$ is the function in Lemma \ref{Lm2} and $R_1$ is given by (\ref{16}).
\end{lemma}

\begin{proof}
By {\bf A4} and {\bf A5}, 
\beq
\label{23}
I_h  (f_d \otimes g_R) \Omega_\omega =
\int {\rm d} t \; h(t)  {\rm e}^{-i t L_\omega} J_0 (f_d \otimes g_{R_1})  \Omega_\omega,
\eeq
since $ t \mapsto {\rm e}^{-i t L_\omega} J_0 (f_d \otimes g_{R_1})  \Omega_\omega \in \cH_\omega$ is
continuous, (\ref{20}) is meaningful and follows from (\ref{19}) by time translation invariance of 
$\omega (.) := (\Omega_\omega , \, . \, \Omega_\omega)$. 
\qed
\end{proof}

Our last preliminary lemma makes use of {\bf A2}: 

\begin{lemma}
\label{Lm4}
Assume {\em ssb} takes place, in the sense that (\ref{13}) holds.  Then, for any $R_0 \in \RR$, 
\begin{eqnarray}
\label{24}
& \kern -1cm
\lim_{R \to \infty}  (\Omega_\omega , [ J_0 (f_d \otimes g_R), \widetilde A_0 ] 
\Omega_\omega)
\nonumber \\ [2mm]
&=
 (\Omega_\omega , [ I_h (f_d \otimes g_{\widetilde R_0}), A_{R_0} ] 
\Omega_\omega) = c \ne 0 ,    
\end{eqnarray}
where
\beq
\label{25}
\widetilde R_0 := 2 R_0 + L +d +1     
\eeq
\noindent
and 
\beq
\label{26}
A_{R_0}  = \frac{1}{| L_{R_0}|} \int_{L_{R_0}} {\rm d}^s \vec x \; \bigl( \sigma_{\vec x} (A_0) - \omega (A_0) \bigr).      
\eeq
Above, $L_{R_0}$ is a $s$-dimensional region of volume 
\beq
\label{27}
| L_{R_0}| = O ( R_0^s)  ,
\eeq
and $\sigma_{\vec x} (A_0) \equiv \pi_\omega ( \sigma_{\vec x} (A_0)) = {\rm e}^{i \vec P \vec x}
\pi_\omega ( A_0) {\rm e}^{-i \vec P \vec x}
$. 
\end{lemma}

\begin{proof}
Applying Lemma \ref{Lm3}, (\ref{20}), and space-translation invariance {\bf A2}, 
\begin{eqnarray}
\label{28}
& \kern -1cm 
\lim_{R \to \infty}  (\Omega_\omega , [ J_0 (f_d \otimes g_R), \widetilde A_0 ] 
\Omega_\omega)
\nonumber \\ [2mm]
& =
 (\Omega_\omega , [ I_h (f_d \otimes g^{\vec x}_{R}), \sigma_{\vec x} (A_0) ] 
\Omega_\omega)    
\end{eqnarray}
for any $\vec x \in \RR^s$, where, by {\bf A5}, 
\beq
\label{29}
g_{R}^{\vec x} (\vec y)  = g_{R} (\vec y-  \vec x).
\eeq
We assume that as a consequence of finite speed of light, the tiny support of $h$ can be taken into account by slightly increasing $d$. By (\ref{28}), (\ref{29})
and Lemma \ref{Lm1}, 
\begin{eqnarray}
\label{30a}
& \lim_{R \to \infty}  (\Omega_\omega , [ J_0 (f_d \otimes g_R), \widetilde A_0 ] 
\Omega_\omega) \qquad 
\nonumber \\ [2mm]
& \qquad =
 (\Omega_\omega , [ I_h (f_d \otimes g_{R_1}), \sigma_{\vec x} (A_0) ] 
\Omega_\omega)    
\end{eqnarray}
as long as 
\beq
\label{30b}
R_1 \ge 2 | \vec x| + L + d +1
\eeq
(\ref{24}) and (\ref{25}) follow from (\ref{30a}), (\ref{30b}) and (\ref{26}). 
\qed
\end{proof}

(\ref{24}) is the starting point for proving our main results in the next section. 

\section{A Goldstone Theorem in Thermal  Field Theory}
\label{III}

As a preliminary to our proof of the Goldstone theorem, we write the self-adjoint operator in (\ref{24}):
\beq
\label{31}
J_0 (f_d \otimes g_{\widetilde R_0})
= \lim_{n \to \infty} \lim_m \sum_{j=1}^m 
\lambda_j' E (\lambda_{j-1}, \lambda_j]
\eeq
where
$\lambda_j' \in (\lambda_{j-1}, \lambda_j]$ 
is arbitrary, 
\beq
-n = \lambda_0 < \lambda_1 < \ldots < \lambda_m = n
\eeq 
and $\lim_m$ is the limit when 
$\max \{ |\lambda_j - \lambda_{j-1}|  \mid j =1, \ldots , m  \} $ tends to zero. This is the spectral theorem (see, e.g., 
\cite{20}, p.~342). We shall abbreviate the double limit in (\ref{31}) by $n,m \to \infty$, and denote the finite sum 
$\sum_{j=1}^m \lambda_j' \cdot E (\lambda_{j - 1} - \lambda_{j}]$
by 
$J_0^{n,m} (f_d \otimes g_{R_1})$. Above, 
\beq
E (\lambda_{j-1}, \lambda_j] = E (-\infty,\lambda_{j}]- E (-\infty,\lambda_{j-1}]
\eeq
are the (right-continuous) spectral projections associated to $J_0$. 
By {\bf A8} et seq., they belong to~$\cR(\cO)$, 
and thus
$ J_0^{n, m} (f_d \otimes g_{\widetilde R_0})\in \cR (\cO)$ for any finite $n, m$. The limit (\ref{31}) is 
understood to be acting on any vector in the domain of $J_0  (f_d \otimes g_{\widetilde R_1})$.

In correspondence to (\ref{24}) and (\ref{31}), we define
\beq
\label{32}
I_h^{n, m} := \int {\rm d}t \; h(t) \tau_{-t} \bigl(J_0^{n, m} (f_d \otimes g_{\widetilde R_0}) \bigr).
\eeq

By definition (\ref{32}), $I_h^{n, m}$ is a smooth element of $\cR(\cO)$ (ignoring once again the fact that one has to
increases the region $\cO$ by a small amount in order to accommodate for the spreading due to the convolution with $h$), and thus, by (iii), belongs to $\cA_S(\cO)$. 
We now apply Corrollary \ref{CorA1} of Appendix A, with $C= I_h^{n, m}$, $A = A_{R_0} = A_{R_0}^*$
(this may be assumed without loss of generality, otherwise we may decompose $A =S +i T$ with $S=S^*$, $T = T^*$. This 
leads to imaginary and real parts of $c$ in (\ref{24}).), and obtain from the Bogoliubov's inequality\eqref{A3}:
\begin{eqnarray}
\label{33}
&\kern -4cm
\frac{1}{\beta} \left| \omega \left( \left[ I_h^{n, m} , A_{R_0} \right] \right) \right|^2 
\nonumber \\ [2mm]
& \le 
\omega \left( \left[ I_h^{n, m} , i \left(\frac{ {\rm d} }{{\rm d}  t } \tau_t (I_h^{n, m})\right)\Bigl|_{t=0}   \right] \right)
\omega \left( A_{R_0}^2 \right) .
\end{eqnarray}

\begin{lemma}
\label{Lm5}
Under the assumptions made, 
\begin{eqnarray}
\label{34a}
&\kern -2cm
\lim_{n,m \to \infty}
\omega \left( \left[ I^{n,m}_h  , 
A_{R_0} \right] \right)
\nonumber \\ [2mm]
& = 
\omega \left( \left[ I_h (f_d \otimes g_{\widetilde R_0}) , 
A_{R_0}
\right] \right),
\end{eqnarray}
and 
\begin{eqnarray}
\label{34b}
&\kern -2cm
\lim_{n,m \to \infty}
\omega \left( \left[ I^{n,m}_h , 
i \left ( \frac{{\rm d}   }{{\rm d}t } \tau_t (I^{n,m}_h ) \right)^* \Bigl|_{ t=0}
\right] \right)
 \\ [2mm]
& = i \omega \left( \left[ I_h (f_d \otimes g_{\widetilde R_0}) , 
\int {\rm d}t \; h'(t) J_0 (f_{d}(\cdot-t)) \otimes g_{\widetilde R_0} ) \right] \right).
\nonumber
\end{eqnarray}
\end{lemma}

\begin{proof}
By {\bf A4}, {\bf A5} and the spectral theorem, 
\begin{eqnarray}
\label{35}
& \kern -2cm
{\rm e}^{i t L_\omega} J_0^{n, m} (f_d \otimes g_{\widetilde R_0}) \Omega_\beta
\quad \quad 
\nonumber \\ [2mm]
& \overrightarrow{\vphantom{X} \scriptstyle n, m \to \infty} \quad  J_0 (f_{d}(\cdot-t)) \otimes g_{\widetilde R_0}) \Omega_\beta
\end{eqnarray}
uniformly in $t \in \RR$. 
Hence, by (\ref{32}),
\begin{eqnarray}
\label{36}
& \kern -1cm
I_h^{n, m} \Omega_\omega = \int {\rm d}t \; h(t) {\rm e}^{i t L_\omega} J_0^{n, m} (f_d \otimes g_{\widetilde R_0}) 
\Omega_\omega
\nonumber \\ [2mm]
&  \qquad \overrightarrow{\vphantom{X} \scriptstyle n, m \to \infty} \quad
\int {\rm d}t \; h(t)  J_0 (f_{d}(\cdot-t)) \otimes g_{\widetilde R_0}) 
\Omega_\omega
\end{eqnarray}
from which (\ref{34a}) follows. Again by (\ref{32}), 
\beq
\label{37}
\left( \frac{ {\rm d}   }{{\rm d} t} \tau_t ( I_h^{n, m}) \right) \Bigl|_{t=0} 
= \int {\rm d}t \; h'(t) \tau_t \left( J_0^{n, m} (f_{d} \otimes g_{\widetilde R_0})\right) . 
\eeq
\noindent
We obtain (\ref{34b}) from (37) by the same argument leading from (\ref{35}) to (\ref{36}). \qed
\end{proof}

We now use assumption {\bf A1}, that $\omega$ is a factor state. 

\begin{lemma}
\label{Lm6}
Let $A, B \in \cA_S$. Then 
\beq
\label{38}
F_{A, B} (\vec x) :=
\Bigl( \omega (A \sigma_{\vec x} (B) ) - \omega (A ) \omega (B)  \Bigr)
\quad \overrightarrow{\vphantom{X} \scriptstyle |\vec x| \to \infty} \quad 0.
\eeq
\end{lemma}

\begin{proof}
This follows from {\bf A1}, {\bf A2} and (\ref{4}) of (iv), see \cite{17}, Theorem 3.2.2.
\qed
\end{proof}

We are now able to state and prove our main result. Assume, with (\ref{38}), that
\beq
\label{39}
F_{A, B} (\vec x) = O ( | \vec x|^{- \delta} ) 
\eeq
as $| \vec x | \to \infty$. Our thermal Goldstone theorem relates the rate of clustering $\delta$ in (\ref{39})
to {\em ssb}:

\begin{theorem}\label{Th1} {\rm (Thermal Goldstone theorem).}  Let a relativistic quantum field theory be defined as a 
$C^*$-algebraic dynamical system $(\cA_S, \omega, \tau)$, satisfying (i)--(iv) as well as {\bf A1}--{\bf A8}. Then, if $s \ge 3$ and if there 
is {\em ssb} as defined by (\ref{13}),  the rate of clustering $\delta$ in (\ref{39}) must satisfy
\beq
\label{40}
\delta \le s-2. 
\eeq
\end{theorem}

\begin{proof} 
By Lemma \ref{Lm4}, (\ref{13}) implies (\ref{24}), (\ref{25}), which, by (\ref{33}) and Lemma \ref{Lm5} leads to the inequality, for any 
$R_0 \in \RR$:
\begin{eqnarray}
\label{41}
& \kern -3cm
\frac{1}{\beta}
\left| \omega \left( \left[ I_h (f_{d} \otimes g_{\widetilde R_0}) , A_{R_0} \right] \right) \right|^2 
\nonumber \\ [2mm]
& \kern -1cm \le 
i \omega \left( A_{R_0}^2 \right) \omega \left( \left[ I_h (f_{d} \otimes g_{\widetilde R_0}), \right. \right. \qquad 
\\
& \qquad \qquad \left. \left.
\int {\rm d} t \; h' (t) J_0 (f_{d}(\cdot-t)\otimes g_{\widetilde R_0})
\right] \right). \nonumber 
\end{eqnarray}
By {\bf A6} (local current conservation) 
\begin{eqnarray}
\label{42}
& \kern -3cm
\left( \int {\rm d} t \; h' (t) J_0 (f_{d}(\cdot-t)) \otimes g_{\widetilde R_0}) \right.
\nonumber \\
& \left. + 
\int {\rm d} t \; h (t) \vec J (f_{d}(\cdot-t)) \otimes \nabla g_{\widetilde R_0}) \right) \Omega_\omega = 0. 
\end{eqnarray}
In (\ref{42}) we applied {\bf A6} to the function 
$f = (h \star f_d) \otimes g_{R_0}$, 
where the asterisk denotes convolution. 

Inserting (\ref{42}) into (\ref{41}) we are led to find a bound to the quantity 
\begin{eqnarray}
\label{43}
&M := i \omega 
\left( \left[ \int {\rm d} t_1 \; h (t_1) \vec J (f_{d}(\cdot-t_1)) \otimes g_{\widetilde R_0}) , \right. \right.
\\
& \qquad \qquad \qquad \left. \left.
\int {\rm d} t_2 \; h (t_2) \vec J (f_{d}(\cdot-t_2)) \otimes \nabla g_{\widetilde R_0})   \right] \right).
\nonumber 
\end{eqnarray}
By (\ref{8}) and (\ref{9e}), (\ref{9f}),
\beq
\label{44}
\frac{\partial g_{\widetilde R_0} }{\partial x_i} = \frac{1}{\widetilde R_0}
\frac{\partial g }{\partial x_i} \left(\frac{\vec x}{ \widetilde R_0 }\right)
\qquad 
i = 1, 2, \ldots, s , 
\eeq
where
\beq
\label{45}
\left(\frac{\partial g}{\partial x_i}\right) \left(\frac{\vec x}{ \widetilde R_0 }\right)
= 0 \; 
\text{\ if  $| \vec x | \le \widetilde R_0$  and $ | \vec x | > \widetilde R_0 + \delta$}
\eeq
and thus
\beq
\label{46}
{\rm supp} \left(\frac{\partial g}{\partial x_i}\right) \left(\frac{\vec x}{ \widetilde R_0 }\right) \subseteq
\Gamma_{\widetilde R_0} 
\eeq
where, by (\ref{45}), 
\beq
\Gamma_{\widetilde R_0} = \{ \vec x \in \RR^s \mid \widetilde R_0 \le | \vec x | \le \widetilde R_0  + \delta \}
\eeq
 is a region of volume
\beq
\label{47}
| \Gamma_{\widetilde R_0} | = | S_{\cS} |
\left[ ( R_0 + \delta)^s - R_0^{s} \right] = O ( \widetilde R_0^{s-1} )
\eeq
with $| S_{\cS} |$ the volume of a $s$-dimensional sphere of unit radius. 
Let $\Gamma_{\widetilde R_0}^{int}$ denote the interior of $\Gamma_{\widetilde R_0}$. We consider the cover (see
Theorem \ref{ThB1} of Appendix B):
\beq
\label{48}
\Gamma_{\widetilde R_0}^{int} = \bigcup_{i \in I} G_i 
\eeq
where
\beq
\label{49}
| I | = O ( \widetilde R_0^{s-1} )
\eeq
and $G_i$ are open hypercubes of side $(1+\epsilon)$, $ 0< \epsilon < 1$, in~$\RR^s$, there being only $O(1)$
such hypercubes along a radius, in accordance to (\ref{46}): in (\ref{49}), $|I|$ is the cardinality of the set~$I$. This
is of course, only one possible choice for the cover (\ref{48}). In correspondence to
the latter, we write now the second term in (\ref{43}) following the theorem on the partition of unity in Appendix B:
let
\beq
\label{50}
\beta_i = \sum_{j \in J} \alpha_j , \qquad {\rm supp} \; \alpha_j \in G_i ,
\eeq
corresponding to {\bf B1}; by {\bf B2} and {\bf B3}, 
$ 0 \le \beta_i \le 1 $ for all~$i$. 

We define 
\beq
\label{51}
 i \in I \mapsto \gamma_i^r := \beta_i \frac{ \partial g (\vec x / \widetilde R_0)}{
 \partial x_r }  \, . 
\eeq
Then, by  (\ref{44})--(\ref{51}):
\begin{eqnarray}
\label{52}
& \kern -3cm
\int {\rm d} t_2 \; h (t_2) \vec J (f_{d}(\cdot-t_2) \otimes \nabla g_{\widetilde R_0}) 
\nonumber \\ [2mm]
& = \frac{1}{ \widetilde R_0 } \sum_{r=1}^s 
\sum_{i \in I} \int {\rm d}t_2 h (t_2)  
J_r (f_{d}(\cdot-t_2) \otimes \gamma_i^r)  . \quad
\end{eqnarray}
By locality {\bf A7b} together with (\ref{52}), we have, for $M$ defined by (\ref{43}):
\begin{eqnarray}
\label{53}
&M = \frac{i }{\widetilde R_0} \sum_{r=1}^s \sum_{i \in I} 
\omega \Bigl( \left[ \int {\rm d}t_1 \, h (t_1) J_0 (f_{d}(\cdot-t_1) \otimes g_{R_i}) , 
\nonumber  \right.  \\ [2mm]
& \qquad \qquad \qquad \left.   \int {\rm d}t_2 \, h (t_2) J_r (f_{d}(\cdot-t_2) \otimes \gamma_i^r) \right] \Bigr) ,
\end{eqnarray}
where
\beq
\label{54}
R_i = O (1) \quad \forall i \in I , 
\eeq
is the minimal length such that ${\rm supp} f_{d}(\cdot-t_1)  \otimes g_{R_i}$ is time-like to 
${\rm supp} f_{d}(\cdot-t_2)  \otimes \gamma_{i}^r $: 
it depends only on $d$, the support of $h$ and 
the diameter of the support of $\gamma_i^r$, which is of order one by our choice of $G_i$ in (\ref{48}). 
By (\ref{49}), (\ref{53}), (\ref{54}) and Assumption {\bf A4}
\beq
\label{55}
0 \le M \le const. \; \widetilde R_0^{s-2} , 
\eeq
where the constant is independent of $\widetilde R_0$. 
By \eqref{14}, \eqref{26}, \eqref{27} and the KMS condition (Assumption {\bf A2})
\beq
\label{56}
\begin{array}{rl}
\omega ( A_{R_0}^2) & = \frac{1}{| L_{R_0} |^2} \int_{L_{R_0}} {\rm d}^3 \vec x
\int_{L_{R_0}} {\rm d}^3 \vec y \; 
\\ [4mm]
& \qquad \qquad \times 
\Bigl( \omega \bigl( A_0 \sigma_{\vec x - \vec y} (A_0) \bigr) - \omega(A_0)^2 \Bigr) 
\\ [4mm]
& \le 
\frac{c}{R_0^s} 
(2 R_0)^{s- \delta} = c R_0^{- \delta} .  
\end{array}
\eeq
Inserting (\ref{41}), (\ref{43}), (\ref{55}) and (\ref{56}) into (\ref{24}) of Lemma \ref{Lm4}, we obtain, with (\ref{25}):
\beq
\label{57}
0 \ne c  \le d \cdot R_0^{s-2- \delta} , 
\eeq
where $d$ is a positive constant independent of $R_0$. (\ref{57})~is true for any $R_0 \in \RR$; taking $R_0 \to \infty$, 
we obtain a contradiction unless \eqref{40} holds. \qed
\end{proof}

We now remark on the restriction to $s\ge3$ in Theorem \ref{Th1}. There are no finite-temperature equilibrium two-point functions (with vanishing chemical potential) for the massless free field for $s=1$ and $s=2$ and nothing is known for interacting theories (see, e.g.,[34], pp. 144 and 151 for a pedagogic discussion). The proof of Theorem \ref{Th1} does not work for $s=1$ (the surface degenerates to a point) and does work for $s=2$ but the result is inconclusive, although \eqref{40}
 correctly predicts a borderline behaviour of the case $s=2$.

For the scalar free field of mass $m$, the two point function corresponding to $F_{A,B}$ in (\ref{38}) is, 
for $s=3$: 
\begin{eqnarray}
\label{58a}
W_\beta (x, m) 
& =& (2 \pi)^{-3} \int {\rm d}^4 p \; \epsilon (p_0) \delta (p^2 - m^2) 
\nonumber \\ 
&& \qquad \qquad \qquad 
\times (1 - {\rm e}^{- \beta p_0})^{-1} {\rm e}^{-i px}  
\\ [2mm]
&=& (2 \pi)^{-3} \int \frac{ {\rm d}^3 \vec p }{ 2 \omega_{\vec p} } \; {\rm e}^{i \vec p \vec x} 
\left( \frac{ {\rm e}^{- i \omega_{\vec p} x_0 } }{ 1 - {\rm e}^{- \beta \omega_{\vec p}  }} + 
\frac{ {\rm e}^{ i \omega_{\vec p} x_0 } }{  {\rm e}^{\beta \omega_{\vec p}  }-1}
\right)   \nonumber
\end{eqnarray}
where 
\beq
\label{58b}
\omega_{\vec p} := (\vec p \, ^2 + m^2)^{1/2} .
\eeq
For $m=0$ the asymptotic behaviour of $W_\beta(x, m)$ for $|x_0| \ll | \vec x | $ is seen from (\ref{58a})
to be the same as that of 
\beq
\label{59}
\int \frac{ {\rm d}^3 \vec p }{ | \vec p |^2 } \; {\rm e}^{i \vec p \vec x} \cong \frac{ 1 }{ | \vec x | }  , 
\eeq
which contrasts with the $\frac{ 1 }{ | \vec x |^2 }$ fall-off in the massless $T=0$ case. (\ref{59}) is also the asymptotic
rate of fall-off of $F_{A, B}(\vec x)$ in (\ref{38}) in the free massless case, and thus the result of Theorem~\ref{Th1} may also be expected to be optimal in thermal (relativistic) quantum field theory, as it is in non-relativistic quantum statistical 
mechanics (see \cite{5, 7}).

We conclude this section with some results and conjectures related to Theorem \ref{Th1}, which help to clarify its significance. The conjectured optimality of Theorem \ref{Th1} suggests the more precise: 

\begin{conjecture}
\label{Con1}
\quad
\begin{itemize}
\item [(i)] in the massless case $(m=0)$ 
\beq
\label{60a}
O \left( | \vec x|^{- \delta} \right) \quad \hbox{with} \quad \delta \le 1 ; 
\eeq
\item [(ii)] in the massive case $(m>0)$ 
\beq
\label{60b}
O \left( | \vec x|^{- \delta} \right) \quad \hbox{with} \quad \delta > 1 . 
\eeq
\end{itemize}
\end{conjecture}

\begin{corollary}
\label{Cor1}
Under Conjecture \ref{Con1}, {\em ssb} of a continuous internal symmetry in thermal relativistic quantum 
field theory with a conserved local current implies the existence of zero mass particles in the theory.
\end{corollary}

Thus,  under Conjecture \ref{Con1} the statement of Goldstone's theorem for $T > 0$ is the same as the 
corresponding one for $T=0$ (see \cite{10a}\cite{10b}).

\begin{remark}
\label{Rm1}
A certain form of slow decay in space-like directions has also been proved in \cite{9} to be necessary 
for the existence of {\em ssb} at $T > 0$ (see (18) of \cite{9}). 
\end{remark}

\begin{remark}
\label{Rm2}
Note that (\ref{60b}) does not assume exponential decay in the $m >0$ case --- as happens in the free field case
(\ref{58a}). This is in agreement with the conjectured behaviour of the damping form factors in \cite{23} (see also
the discussion in \cite{1}). 
\end{remark}

Is there a spectral theoretic statement related to $(\ref{60a})$ and $(\ref{60b})$? Since
the spectrum $\sigma (L_\omega)$ of $L_\omega$ is the whole real line,
\beq
\label{62}
\sigma (L_\omega) = \RR , 
\eeq
this question has no obvious answer. However, if $\Omega_\omega$ is the unique (up to a phase) 
normalised eigenvector of $L_\omega$ with eigenvalue $0$, then $\omega$ is a factor state \cite{24} and one has the following result
(see \cite{1, 29}): 

\begin{theorem} Let $\Omega_\omega$ be as above, and $P^+$ denote the projection onto the strictly positive part of $\sigma(L_\omega)$. Assume there exist positive constants $\delta > 0$ and $C_1 (\cO)$
such that 
\beq
\label{63}
 \| {\rm e}^{- \lambda L_\omega} P^+ \pi_\omega (A) \Omega_\omega \|
 \le C_1 (\cO) \lambda^{-   \delta} \| A \|    
 \eeq
for all $A \in \cA ( \cO)$. Consider now two space-like separated space-time regions $\cO_1$, $\cO_2$, which can be embedded into~$\cO$
by translation and such that $\cO_1 + r e \subset \cO_2'$, $ r \gg \beta$; 
then, for all $A \in \cA ( \cO_1)$, and all $B \in \cA (\cO_2)$ 
\beq
\label{64}
| \omega (BA) - \omega (B) \omega (A) | \le C_2 \, r^{- 2 \delta} \| A \| \| B \|  \, . 
\eeq
The constant $C_2 (\beta, \cO) \in \RR^+$ may depend on the temperature $T = \beta^{-1}$ and the size of the region
$\cO$, but is independent of $r$, $A$, and $B$. 
\end{theorem}

As remarked in \cite{1}, from explicit calculations one expects that $\delta = 1/ 2$ for 
free massless bosons in $3+1$ space-time dimensions, 
and thus the exponent on the r.h.s.~of (\ref{64}) is optimal due to (\ref{59}). 

It is interesting that, in the massive case, for $T =0$, exponential decay on the r.h.s.~of~(\ref{64}) follows from 
the spectral gap in $H_\omega> 0$, i.e., exponential decay in $\lambda$ of 
\beq
\| {\rm e}^{- \lambda H_\omega} \pi_\omega (A) \Omega_\omega \| , 
\eeq
by the cluster theorem \cite{25}, while, for $T > 0$, 
sufficiently fast polynomial decay of correlations --- (\ref{64}), with $\delta > 1/2 $ --- equally follows from 
sufficiently fast decay of 
\beq
\| {\rm e}^{- \lambda L_\omega} P^+ \pi_\omega (A) \Omega_\omega \| , 
\eeq
--- (\ref{63}) with $\delta > 1/2$ --- if (\ref{60b}) is correct. It is to be remarked that (79) is related (see \cite{28})
to the Buchholz-Wichmann nuclearity property \cite{27}.

\section{Discussion and Outlook}
\label{IV}

In this paper we have shown that a Goldstone theorem may be proved in thermal quantum field theory, relating
{\em ssb} to the space-like decay of the two-point function (Theorem III.3 of Section~\ref{III}). 
Since the limiting behaviour (\ref{40}) of Theorem III.1 agrees with that of the massless free field theory (\ref{59}), 
we were led to the conjecture that the theorem may be optimal, as occurs in non-relativistic quantum statistical mechanics, leading to a sharp distinction~(\ref{60a}),~(\ref{60b}) between massive and massless thermal rqft. The latter is found by examining the rate of fall-off of the two-point function only in space-like directions. If this conjecture is correct, 
Corollary III.5 provides a statement of Goldstone's theorem for $T > 0$, which is quite analogous to the one for $T=0$
(see \cite{10a, 10b}).

We have chosen to set our scale large only as far as space-like distances are concerned. As remarked in 
\cite{26}, this may be appropriate for discussing global issues like superselection sectors, statistics and symmetries. But there remains scattering theory with the associated notions of particles and infraparticles, and there large time-like distances are crucial. 
Thus, if one is really concerned with unravelling the concept of particle in thermal rqft, the approach of Bros and Buchholz ([9,23]) is the most natural one.  However, time-like decay as~$|t|^{-3/2}$ for $\vec x = \vec v t$ (which follows from (\ref{58a})) leads, together with the assumption of a sharp dispersion law, to the famous Narnhofer-Requardt-Thirring theorem \cite{29}, according to which there is no interaction. 
We refer to \cite{23} (see Sect.~d., p.~518) for a lucid discussion of possible ways out of this dilemma, but the matter still remains under discussion.

A relevant open problem is a purely algebraic version of the Goldstone theorem in the case of positive temperature, in analogy to what was accomplished in ref. [36] for the ground state. It should also be remarked that domain problems such as the one pertaining to assumption A4 have been solved in [36], without the need of any assumption, in a very ingenious way (see (3.6) et seq), but we were unable to do the same here. In addition, nonconserved currents, successfully dealt with in [36], remains an open problem for $T>0$. Finally, {\em ssb} of Lorentz and Galilei symmetries has been studied by a different method in ref. [37], where references to related work by Requardt are to be found.
 
A different but fundamental set of issues related to time-like clustering, not
mentioned in \cite{26}, concerns stability.  The time-like cluster property (also called  mixing property \cite{24}
\beq
\label{65}
\lim_{t \to \infty} 
\Bigl( \omega (A \tau_t (B)) - \omega (A)\omega (B) \Bigr) =0
\eeq
implies, for $T > 0$, the dynamic stability condition of Haag, Kastler and Trych-Pohlmeyer \cite{30}
\beq
\label{66}
\lim_{T \to \infty} 
\int_{-T}^T {\rm d} t  \;  \omega ([ A,  \tau_t (B)])  =0
\eeq
(see \cite{16}, Vol.~2, Theorem 5.4.12, pg.~165). Although (\ref{65}) has been proved for the ground state of relativistic quantum field theories \cite{31}, it is still open for thermal KMS states, although a similar property has been proved for a weakly dense set of (in general non-KMS) states [33]. Proof of (\ref{65}) for KMS states would imply the property of return to equilibrium \cite{24}, as well as the dynamic stability condition~(\ref{66}), both quite deep, and in general, hard to prove (see \cite{24} for references). 

\bigskip

{\em Acknowledgements.}
This paper originates from discussion during our stay at the Erwin Schr\"odinger Institute (ESI) 
Vienna, from 1--14 June 2009. We are grateful to Prof.~J.~Yngvason for making our visit possible, as well as to 
Heide Narnhofer and Geoffrey Sewell for several enjoyable discussions on related matters. We also thank Heide Narnhofer for remarks on the manuscript, as well as Manfred Requardt for letting us know of ref [37]. 
  
\section{Appendix A}
\label{V}

\begin{theorem}
\label{ThA1} Let $\cA$ be a $C^*$-algebra and $\omega$ a state on $\cA$ satisfying the KMS condition~\eqref{5} 
w.r.t.~a group of norm-continuous automorphisms $\{ \tau_t \}_{t \in \RR}$. 
Let $ A \in \cA$ and $C \in \cA$ be both of the form 
\beq
\label{A1}
C = \int {\rm d} t \; h (t) \tau_t (B)
\eeq
with some $B \in \cA$ and 
\beq
\label{A2}
\hat h \in \cD = C_0^\infty (\RR),
\eeq
where $\hat h$ denotes Fourier transform of $h$. Then 
\beq
\label{A3}
\frac{2}{\beta} \, | \omega ( [ C, A^*])|^2  
\le \omega \left( \left[ C , i \left( \frac{{\rm d} }{{\rm d} t} \tau_t (C) \right) \Bigl|_{t=0} \right] \right) \cdot 
\omega ( \{ A, A^* \} ) ,
\eeq
where $\{ A, B \} := AB + BA$.
\end{theorem} 

Inequality (\ref{A3}) ({\em Bogoliubov's inequality}) may be extended to all $A \in \cA$ and to those $C$ of the form 
\beq
\label{A4}
C = \int {\rm d} t \; g (t) \tau_t (B),
\eeq
where $g \in C^\infty (\RR)$ is such that, given any $\epsilon > 0$, there exists $h$ satisfying (\ref{A2})
such that
\beq
\label{A5}
 \int {\rm d} t \; | h' (t) -g'(t) | < \epsilon 
\eeq
and 
\beq
\label{A5b}
 \int {\rm d} t \; | h (t) -g(t) | < \epsilon. 
\eeq

\goodbreak
\begin{proof}
See \cite{13} and \cite{16}, Vol.~II, pg.~333. Norm-continuity of the time-translation automorphisms was not explicitly stated in \cite{13}, 
but is used to extend the result from $A$ in the class (\ref{A1}), (\ref{A2}) to the whole of~$\cA$ by density (see
\cite{19}, Theorem 4.8). The extension to (\ref{A4}) was not mentioned in \cite{13}, but follows from (\ref{A1}), (\ref{A2}), (\ref{A3}), (\ref{A4}) 
and (\ref{A5}), together with 
\beq
\label{A6a}
\frac{{\rm d} }{{\rm d} t} \tau_t (C) \Bigl|_{t=0} = \int {\rm d}t \; g' (t) \tau_t (B) \eeq
for $C$ of the form (\ref{A4}), and 
\beq
\frac{{\rm d} }{{\rm d} t} \tau_t (C) \Bigl|_{t=0} = \int {\rm d}t \; h' (t) \tau_t (B) 
\eeq
for $C$ of the form (\ref{A1}).
\qed
\end{proof}

\begin{corollary}
\label{CorA1}
Let $A \in \cA_S$ (see Section~\ref{II}) and $C$ be of the form (\ref{A4}), with $B \in \cA_S$ and 
\beq
\label{A7}
g \in C^\infty_0 (\RR).
\eeq
Then, if $\omega$ is a state on $\cA_S$ satisfying the KMS condition, Condition {\bf A3} holds.
\end{corollary}

\begin{proof}
Since $h \in \cS(\RR)$ by (\ref{A2}), given $g$ satisfying (\ref{A7}), we may choose the `tail to infinity' in $h$
appropriately so that \eqref{A5} and \eqref{A5b} hold. 
\qed
\end{proof}

\begin{remark}
\label{ReA1}
We use \eqref{A3} in the main text under conditions of Corollary \ref{CorA1}. Since conditions (\ref{A2}) and (\ref{A7})
are mutually excludent by the Paley-Wiener theorem (see e.g. \cite{20}, Exercise 8 of Chap. 10), the  density 
argument in Theorem \ref{ThA1} is important for the application we make of \eqref{A3} in Section~\ref{III}. 
\end{remark}

\begin{remark}
\label{ReA2}
For the proof of positivity of the middle term in (85) and other questions related to the Bogoliubov scalar product, see (\cite{16}, Vol.~II, pg.~334).
For some inequalities in statistical mechanics for $W^*$-systems, see~\cite{21}.
\end{remark}

\section{Appendix B}
\label{VI}

We state here, for the reader's convenience, the theorem (partition of unity) used in Theorem~\ref{Th1} of Section~\ref{III}:

\begin{theorem}
\label{ThB1}
{\rm (see \cite{22}, Theorem, p.~61, Chap.~I, Sect.~12).} Let $G$ be an open set of $\RR^n$, and let a family of open
sets $\{ G_i \mid i \in I \}$ cover $G$, i.e., $G = \bigcup_{i \in I} G_i $. Then there exists a system of functions 
$\{ \alpha_j (x) \mid j \in J \}$ of $C_0^\infty (\RR^n)$ such that 
\begin{itemize}
\item [{\bf B1}] for every $j \in J$, $\supp (\alpha_j) $ is contained in some $G_i$;
\item [{\bf B2}] for every $j \in J$, the function $ \alpha_j$ satisfies $0 \le \alpha_j (x) \le 1$ for all $x \in \RR^n$;
\item [{\bf B3}] $ \sum_{j \in J}\alpha_j =1 $ for $x \in G$.
\end{itemize}
\end{theorem}


\begin{thebibliography}{Bor}

\bibitem{1} C.~J\"akel, {\em Thermal quantum field theory},  published in: Encyclopedia of Mathematical 
Physics, Elsevier, Editors: Jean-Pierre Francoise, Gregory Naber, Tsou 
Sheung Tsun, ACADEMIC PRESS (2006) ISBN: 0-12-512660-3. 
\bibitem{2} N.P.~Landsman and Ch.G.~van Weert, {\em Real- and imaginary-time field theory at finite temperature and density}, Phys. Rep. 145, 141--249 (1987). 
\bibitem{3} D. Buchholz, I. Ojima and H. Roos, {\em Thermodynamic properties of non-equilibrium states in quantum field theory}, Ann. Phys. 297, 219--242 (2002).
\bibitem{4} C. Dappiaggi, K. Fredenhagen and N. Pinamonti, {\em Stable cosmological models driven by a free quantum scalar field}, Phys. Rev. D 77 (2008) 104015.
\bibitem{5} W.F. Wreszinski,  {\em Charges and symmetries in quantum theories without locality}, 
Fortschritte der Physik 35, 379--413 (1987).
\bibitem{6} L.J. Landau, J. F. Perez and W.F. Wreszinski, {\em Energy gap, clustering and the Goldstone theorem in quantum statistical mechanics}, J. Stat. Phys. 26, 755--766 (1981).
\bibitem{7} P.A. Martin, {\em A remark on the Goldstone theorem in statistical mechanics}, Nuovo Cimento 68B, 302--314 (1982)
\bibitem{8} D. Mermin and H. Wagner, {\em Absence of ferromagnetism or antiferromagnetism in one- or two-dimensional isotropic Heisenberg models}, Phys. Rev. Lett. 17, 1133--1136 (1966); see also D. Ruelle, {\em Statistical Mechanics: Rigorous Results}, Addison-Wesely Publishing Co., Inc. (1969) and (1989).
\bibitem{9} J. Bros and D. Buchholz, {The unmasking of thermal Goldstone bosons}, 
Phys. Rev. D58, 125012 (1998). 
\bibitem{10a} D. Kastler, D.W. Robinson and J.A. Swieca, {\em Conserved currents and associated symmetries; Goldstone's theorem}, 
Commun. Math. Phys. 2, 108--120 (1966).
\bibitem{10b} H. Ezawa and J.A. Swieca, {\em Spontaneous breakdown of symmetries and zero-mass states}, Commun. Math. Phys. 5, 330--336 (1967).
\bibitem{11} F.J. Dyson, E. Lieb, and B. Simon, {\em Phase transitions in quantum spin systems with isotropic and 
non-isotropic interactions}, J. Stat. Phys. 18, 335--383 (1978).
\bibitem{12} J.A. Swieca, {\em Goldstone's theorem and related topics}, Carg\'ese Lectures in Physics, Ed. D. Kastler, Gordon and Breach (1970). 
\bibitem{13} J.C. Garrison and J. Wong, {\em Bogoliubov inequalities for infinite systems},
Commun. Math. Phys. 26, 1--5 (1972). 
\bibitem{14} J. Naudts, A. Verbeure and R. Weder, {\em Linear response theory and kms condition}, Commun. Math. Phys. 44, 87--99 (1975).
\bibitem{15}  G. Roepstorff, {\em Correlation inequalities in quantum statistical mechanics and their
application in the Kondo problem}, Commun. Math. Phys. 46, 253--262 (1976).
\bibitem{16} O. Bratteli and D.W. Robinson, {\em Operator Algebras and Quantum Statistical Mechanics Vol. I,II}, Springer-Verlag, New York-Heidelberg-Berlin (1981). 
\bibitem{17} R. Haag, {\em Local Quantum Physics: Fields, Particles, Algebras}, Springer-Verlag, Berlin-Heidelberg-New York, 2nd revised and enlarged edition (1996). 
\bibitem{18} R.V. Kadison and J.R. Ringrose, {\em
Fundamentals of the Theory of Operator Algebras}, 
Vol. 2, Academic Press (1986) ISBN	0123933021.
\bibitem{19} N.M. Hugenholtz, {\em States and representations in statistical mechanics}, published in: Mathematics of contemporary physics, Streater (Ed.), p. 145., Academic Press, London, (1972).
\bibitem{20}  Ph. Blanchard and E. Bruening, {\em Mathematical Methods in Physics},
Birkh\"auser Boston (2002) ISBN-10: 0817642285, ISBN-13: 978-0817642280.
\bibitem{21}  Ruskai, M.B., Inequalities for traces on von Neumann algebras, Commun. Math. Phys. 26, 280--289 (1972).
\bibitem{22} K.~Yoshida,  {\em Functional Analysis and Its Applications}, Springer-Verlag, New York (1971). 
\bibitem{23} J. Bros and D. Buchholz, {\em Axiomatic analyticity properties and representations 
of particles in thermal quantum field theory}, Ann. Inst. H. Poincar\'e 64, 
495--521 (1996). 
\bibitem{24} C.A. Pillet, {\em Quantum Dynamical Systems}, in S.Attal, A. Joye and C.A. Pillet, eds, Open Quantum Systems I, the Hamiltonian Approach, vol. 1880 of Lecture Notes in Mathematics, Springer-Verlag, New York 2006.
\bibitem{25} K. Fredenhagen,  {\em A remark on the cluster theorem}, Commun. Math. Phys. 97, 461--463 (1985). 
\bibitem{26} R. Brunetti and K. Fredenhagen, {\em Algebraic Approach to Quantum Field Theory}, 
published in:  Encyclopedia of Mathematical 
Physics, Elsevier, Editors: Jean-Pierre Francoise, Gregory Naber, Tsou 
Sheung Tsun, ACADEMIC PRESS (2006) ISBN: 0-12-512660-3. 
\bibitem{27} D. Buchholz and E. Wichmann, {\em Causal independence and the energy level density of 
states in local quantum field theory}, Commun. Math. Phys. 106, 321--344 (1986).
\bibitem{28} C.D. J\"akel, {\em Decay of spatial correlations in thermal states}, Ann. Inst. Henri Poincar\'e 69, 425 (1998).
\bibitem{29} H. Narnhofer, M. Requardt and W. Thirring, {\em Quasi-particles at finite temperatures}, 
Commun. Math. Phys. 92, 247--268 (1983). 
\bibitem{30} R. Haag, D. Kastler, and E.B. Trych-Pohlmeyer, {\em Stability and equilibrium 
states}, Commun. Math. Phys. 38, 173--193 (1974), see also:
R. Haag and E.B. Trych-Pohlmeyer, {\em Stability properties of equilibrium 
states}, Commun. Math. Phys. 56, 213--224 (1977). 
\bibitem{31} D. Maison, {\em Eine Bemerkung zu Clustereigenschaften}, Comm. Math. Phys. 10, 48--51 (1968).
\bibitem{32} C.D., J\"akel, H., Narnhofer, and W.F., Wreszinski, {\em On the mixing property for a class of states of 
relativistic quantum fields}. 
\bibitem{33} S. A. Fulling and S. N. M. Ruijsenaars, {\em Temperature, periodicity and horizons}, Phys. Rep. 152, 135-176 (1987).
\bibitem{34} M. Requardt, {\em Symmetry Conservation and Integrals over Local Charge Densities in Quantum Field Theory},
Comm. Math. Phys. 50, 259-263 (1976).
\bibitem{35} D. Buchholz, S. Doplicher, R. Longo and J.E. Roberts, {\em A New Look at Goldstone's Theorem}, Rev. Math. Phys. (Special Issue) 49-83 (1992).
\bibitem{36} M. Requardt, {\em Spontaneous Symmetry Breaking of Lorentz and (Galilei) Boosts in (Relativistic) Many Body Systems}, arXiv 0805.3022 (2008).
\end{thebibliography}
\end{document}